\documentclass[a4paper]{article}
\usepackage{graphicx} 
\usepackage{amsmath, amsfonts, amssymb, amscd, bm}
\usepackage{amsthm}
\usepackage{verbatim}
\usepackage{enumerate}
\usepackage{hyperref}
\usepackage{url}

\newtheorem{theorem}{Theorem}[section]
\newtheorem{corollary}{Corollary}[theorem]
\newtheorem{lemma}{Lemma}[section]
\newtheorem{definition}{Definition}[section]

\theoremstyle{remark}
\newtheorem{remark}{Remark}[section]

\begin{document}

\title{Time consistency and moving horizons in risk measures}
\author{Samuel N. Cohen \\ University of Adelaide \\ samuel.cohen@adelaide.edu.au 
\and Robert J. Elliott\thanks{Robert Elliott wishes to thank the Australian Research Council for
support.} \\ University of Adelaide and University of Calgary\\ relliott@ucalgary.ca}

\maketitle
\begin{abstract}
We consider portfolio selection when decisions based on a dynamic risk measure are affected by the use of a moving horizon, and the possible inconsistencies that this creates. By giving a formal treatment of time consistency which is independent of Bellman's equations, we show that there is a new sense in which these decisions can be seen as consistent.
%\keywords{dynamic risk measure \and time inconsistency \and moving horizon}
%\subclass{60H10 \and 93E20}
\end{abstract}

\section{Introduction}
Risk is an active area of study. The management of uncertain outcomes, and decision making in this context, is of considerable importance. Much recent research has focussed around properties of `coherent risk measures', as first discussed in \cite{Artzner1999}, and `convex risk measures', as defined by \cite{Follmer2002} and \cite{Frittelli2002}. These are functionals $\rho: L^1(\mathcal{F}_T)\to \mathbb{R}$, where $T$ is some future time, and $L^1(\mathcal{F}_T)$ is the space of integrable $\mathcal{F}_T$-measurable random variables. In the convex case, it is assumed that these functionals satisfy three assumptions, namely:
 \begin{enumerate}
 \item Monotonicity: $X\geq Y, \mathbb{P}\text{-a.s.} \Rightarrow \rho(X)\leq \rho(Y)$,
 \item Translation invariance: $\rho(X+c) = \rho(X)-c$ for all $c\in\mathbb{R}$,
 \item Convexity: $\rho(\lambda X + (1-\lambda)Y) \leq \lambda\rho(X)+(1-\lambda)\rho(Y)$ for all $\lambda\in[0,1]$.
 \end{enumerate}
  
 One significant flaw with these risk measures is that they are essentially static -- they consider only one random outcome, and do not model the development of information through time. Simply applying these risk measures to a multiple-period problem is insufficient, as there is no guarantee that they will lead to time-consistent decision making. In particular, there is no guarantee that Bellman's principle will be satisifed. Concrete examples of this can be found in \cite{Kang2006} and \cite{Artzner2007}.
 
 More recently, Artzner et al. \cite{Artzner2007} discussed how a particular expression of Bellman's principle is equivalent to a recursivity property of the risk measures, namely if $\rho_t(X)$ denotes the risk of $X$ as considered at time $t$, then for any $s<t$, we have $\rho_s(X) = \rho_s(-\rho_t(X))$.  In \cite{Kloppel2007}, an equivalent property, (given translation invariance), is considered. Specifically, in \cite{Kloppel2007} a type of inter-temporal monotonicity is assumed, that is, for any times $s<t$, $\rho_t(X)\geq \rho_t(Y)$ $\mathbb{P}$-a.s. implies $\rho_s(X)\geq \rho_s(Y)$ $\mathbb{P}$-a.s. In this paper, we show that a version of this monotonicity condition is equivalent to a general form of Bellman's principle; see Theorem \ref{thm:equivtimeconsistent}. 
 
  Much has been written on \emph{dynamic risk measures}, that is, risk measures where a recursivity property is satisfied. See, for example, \cite{Rosazza2006}, \cite{Barrieu2004}, \cite{Delbaen2008}. Similarly, a theory of `time-consistent nonlinear expectations' has been developed. See particularly \cite{Peng1997} and the references therein. These satisfy assumptions very similar to those of dynamic risk measures, the main difference being a sign change in each of the three assumptions above. To construct these functionals, a common tool is the theory of Backward Stochastic Differential Equations (BSDEs), and it is known that all nonlinear  expectations, (satisfying some constraints), can be expressed as solutions of BSDEs, (see \cite{Coquet2002} and \cite{Hu2008} in the continuous time case, \cite{Cohen2008c} and \cite{Cohen2009a} in the discrete time case). 
 
To apply these methods, one must typically fix a distant point $T$ in the future, (possibly infinitely distant), at which all payoffs will be realised. Alternatively, as for example in \cite{Peng2004} or \cite{Bion-Nadal2008}, one can generalise the risk measures to operators $\rho_{\sigma, \tau}:L^\infty(\mathcal{F}_\sigma) \to L^\infty(\mathcal{F}_\tau)$ where $\sigma\leq \tau$ are stopping times. If we assume that $\tau\leq T$ for some fixed $T$, we can then replace $\tau$ with $T$ throughout, by the property $\rho_{\sigma,\tau}(X)=\rho_{\sigma,T}(X)$ for all $X\in L^\infty(\mathcal{F}_\tau)$. 

In many investment applications, predicting even the distribution of extremely long-term behaviour is almost impossible. One might hope to use a shorter-dated moving horizon, where the portfolio value at some fixed time into the future, (say, one-year from the present), is considered, but this horizon is allowed to move forward as time progresses. That is, the risk is calculated based on the portfolio value a short time in the future, rather than at the terminal time $T$. Hence, if $V_{t}$ is the portfolio value at time $t$, our risk at time $t$ is measured by $\rho_t(V_{t+m})$, where $m$ is the horizon distance. As $V_{s+m}\neq V_{t+m}$ in general, it is clear that the recursivity properties imply no relationship between $\rho_s(V_{s+m})$ and $\rho_s(-\rho_t(V_{t+m}))$, and an approach to time-consistency based on the recursivity property of $\rho$ is insufficient. In this paper, we discuss what consistency properties remain under such a regime.

These and similar questions have also led to the forward-performance approach of Musiela and Zariphopolou and others (for example, \cite{Musiela2009}, \cite{Musiela2010}, \cite{Zitkovic2009}). Here a self-generation property is used to ensure that, under optimal behaviour, $\rho_s(V_{t})$ is independent of $t$ for $t\geq s$, and so no problem arises. However, this requires that $\rho$ is of a very special type, and depends on the model used of the market. We here consider the consequences of simply assuming that $\rho$ satisfies the standard recursivity property.

This work is also motivated by applications in economic regulation. In many risk management settings the `risk' is calculated over some finite horizon, to ensure it does not exceed certain bounds. For example, in the Basel II Banking accords, regulators calculate a ten-day 99\%-value-at-risk for market risk, and a one-year 99.9\%-value-at-risk for credit and operational risks. See \cite{Hull2009} for more details. Even if these risks are calculated using a dynamic risk measure, (which, as is well known, value-at-risk is not), the moving horizon will introduce inconsistencies into the analysis.

Time inconsistent problems have been classically studied in economics, for example, in the works of \cite{Goldman1980}, \cite{Peleg1973}, \cite{Pollak1968} and   \cite{Strotz1955} or more recently in \cite{Ekeland2008} and \cite{Bjork2009}. The approach used in these papers is based on solving an intertemporal game. In particular, the `optimal' strategy is selected subject to the requirement that there will be no benefit from deviating from it at any point in the future. Clearly, determining such strategies generally requires explicit consideration of future behaviour. This approach to choosing strategies is not of key interest for the moving horizon problem. The reason for this is that if one were to consider the actions one will take tomorrow, one would have to consider behaviour up to tomorrow's horizon, that is, one day further than the horizon considered today. In this case, one may as well consider this more distant horizon directly. A recursive argument then shows that this would result in the horizon being extended into the distant future, and therefore, the moving horizon problem would essentially disappear.

For this reason, we consider the situation where decisions are made in a completely na{\"\i}ve manner, without regard for future behaviour. Our question is whether this approach will yield time-consistent policies, which will clearly depend on those policies available and the values assigned to them. We shall show that, for a simple dynamic investment problem, decision making with a moving horizon is not time-consistent in general. We shall then show that there exists a modified version of time-consistency which is satisfied, given certain assumptions on the possible policy space.

For simplicity, we shall work in a discrete time setting. The continuous time setting is conceptually similar, and we expect that many of the results obtained will carry over, with appropriate technical modification. However, there are significant difficulties in working with moving horizons in continuous time, some of which are explored in \cite{Choulli2009}.

We proceed by first considering the fundamental notions of time-consistency, and derive an appropriate variant for our problem. In Section \ref{sec:movhorizintro} we then formally introduce the particular problem under consideration, and show that the classical requirements of time-consistency are not satisfied; however, the requirements for our modified concept are.

\section{Time-consistency and policies}

We now present a general definition of time-consistency, which is essentially a formalisation of Bellman's Principle of Optimality. While taking Bellman's Principle as a useful basis for a definition of time consistency, we shall not assume that the value function is the solution of Bellman's equation and, hence, the problems considered may not be time-consistent.

In general, we assume that there is a set of allowable policies $\mathcal{U}$, which are adapted processes taking values in some metric space $\mathbb{U}$. They are selected to optimise some value function $\mathcal{V}$, which is in general a family of maps
\[\mathcal{V}_t:\mathcal{U}\to L^1(\mathcal{F}_t),\quad t\in\{0,1,...T\}.\]
For simplicity, we take higher values of $\mathcal{V}$ as better than lower. 

\begin{definition}\label{def:viability}

Let $X\in\mathcal{U}$ be a policy. We define the \emph{conditional policy space} at $t$ given past policy $X$,
\begin{equation}\label{eq:conditionalpolicyspace}
\mathcal{U}|_t^{X}=\{X'\in \mathcal{U}: X_s'=X_s \quad \mathbb{P}-a.s. \text{ for all } s<t\}.
\end{equation}
Note that $\mathcal{U}|_0^{X} = \mathcal{U}$ for all $X$.

Let $\{X^t\}$ be a collection containing a policy choice $X^t\in\mathcal{U}$ for each time $t\leq T$. Let $\hat X$ denote that policy which is eventually chosen, that is $\hat X_u:= X^u_u$. Then this policy choice is \emph{viable} if, for every $s<t$,
\[X^t\in \mathcal{U}|_s^{X^s}\]
or equivalently
\[X^t\in \mathcal{U}|_{t}^{\hat X} = \bigcap_{s<t} \mathcal{U}|_{s}^{X^s} = \mathcal{U}|_t^{X^{t-1}}.\]
\end{definition}

To ensure that in different states of the world different decisions can be independently made, we have the following property.

\begin{definition}
We say the conditional policy space satisfies the \emph{pasting property} if for any past policy $\hat X\in\mathcal{U}$,
\begin{equation}\label{eq:pastingproperty}
I_A X+ I_{A^c}X' \in \mathcal{U}|_t^{\hat X} \quad \text{for all}\quad X,X' \in \mathcal{U}|_t^{\hat X}, A\in\mathcal{F}_t.
\end{equation}
We say the value function $\mathcal{V}$ satisfies the zero-one law if 
\begin{equation}\label{eq:zeroonelaw}
\mathcal{V}_t(I_A X + I_{A^c}X') = I_A \mathcal{V}_t(X) + I_{A^c}\mathcal{V}_t(X') 
\end{equation}
for all $X,X'\in\mathcal{U}|_t^{\hat X}$, $A\in\mathcal{F}_t$.
\end{definition}

 Intuitively, we think of $X^t$ as the policy which an investor intends to pursue, when making a selection at time $t$. A collection being viable ensures that $\hat X$, the policy that is finally chosen,  does not involve an investor attempting to change their past actions at any time.
 
 Note that this definition requires that the past policy is matched both in the observed past and in all possible other pasts (that is, for all $\omega$). This requirement is needed to ensure that switching, at time $t$, from one policy $X$ to another in $\mathcal{U}|_t^{\hat X}$ results in a policy which is in $\mathcal{U}$. 
 
  The following result ensures that $\hat X$ is in fact a policy, that is, $\hat X \in \mathcal{U}$.
\begin{lemma}\label{lem:XtnearT}
If $\{X^t\}$ is a viable policy choice, then $X^T=\hat X$. Hence $\hat X \in \mathcal{U}$.
\end{lemma}

We now give simple conditions under which our problem has a solution. Our main focus is not on deriving conditions for the solution to exist, but on exploring the implications of the solution for time-consistency; therefore, the restrictive nature of these conditions is not a major concern. It is easy to see that our main results all have appropriate modifications to more general settings whenever the existence of optimal policies is given.

\begin{definition}
In general, we shall say that our problem is \emph{standard} if 
\begin{itemize}
\item $\mathcal{U}$ is a compact subset of adapted processes on $\mathbb{U}$, with induced metric $d(X, X') = \sum_s E[d_{\mathbb{U}}(X_s, X_s')]$. For any $X, X'\in\mathcal{U}$, this metric satisfies
\[d(X, X') = 0 \quad\text{ if and only if }\quad X_s= X'_s \quad \mathbb{P}-a.s. \text{ for all } s.\]
Note that we do not assume that $\mathbb{P}$-almost sure convergence is metrizable. A simple example is a compact subset of adapted processes taking values $X_t\in L^1(\mathbb{R}^N;\mathcal{F}_t)$ for all $t$.

\item For all $t$, the value function $\mathcal{V}_t:\mathcal{U}\to L^1(\mathcal{F}_t)$ is lower semicontinuous under the metric topology, that is, if $X^n\to X^\infty$ and $\mathcal{V}_t(X^n)\leq \mathcal{V}_t(X^{n+1})$ $\mathbb{P}$-a.s. for all $n$, then $\lim_{n\to\infty} \mathcal{V}_t(X^n) = \mathcal{V}_t(X^\infty)$.

\item For all $t$, $\mathcal{U}|_t^{\hat X}$ as defined by (\ref{eq:conditionalpolicyspace}) satisfies the pasting property (\ref{eq:pastingproperty}) and $\mathcal{V}_t$ satisfies the zero-one law (\ref{eq:zeroonelaw}).
\end{itemize}
\end{definition}

\begin{lemma} \label{lemma:Ucondcompact}
If $\mathcal{U}$ is compact, then $\mathcal{U}|_t^{\hat X}$ is compact for all times $t$ and all policies $\hat X$. 
\end{lemma}

\begin{proof}

For any sequence $\{X^n\}$ in $\mathcal{U}|_t^{\hat X}$, we know that $X^n_u = \hat X_u$ for all $u< t$. As $\mathcal{U}$ is compact, there exists a convergent subsequence of $X^n$. This subsequence has a modification in $\mathcal{U}|_t^{\hat X}$ by the pasting property, and as we are in discrete time these modifications are indistinguishable. Therefore, $\mathcal{U}|_t^{\hat X}$ is sequentially compact and hence compact.
\end{proof}

\begin{theorem} \label{thm:Optimalpolicyexists}
For a standard problem, for all $t$ and any past policy $\hat X$, there exists a policy $X^t\in\mathcal{U}|_t^{\hat X}$ such that for any $X\in\mathcal{U}|_t^{\hat X}$
\[\mathcal{V}_t(X^t)\geq \mathcal{V}_t(X),\qquad \text{a.s.}\]

We shall say that the policy $X^t$ given by Theorem \ref{thm:Optimalpolicyexists} \emph{uniformly maximises} $\mathcal{V}_t$ on $\mathcal{U}|_t^{\hat X}$.
\end{theorem}
\begin{proof}
This is classical result from the assumptions of a standard problem.
\end{proof}

\begin{definition} \label{def:timeconsistent}
Let $\{X^t\}$ be a viable collection containing a policy choice for each time $t$. This generates a realised policy $\hat X$, defined by $\hat X_t = X^t_t=X^T_t$ for all $t$. This collection is called \emph{optimal}  if
\begin{enumerate}[(i)]
\item for any $t$, the policy $X^t$ uniformly maximises $\mathcal{V}_t(X)$ for $X\in\mathcal{U}|_{t}^{\hat X}$,
\end{enumerate}
and \emph{time consistent} if
\begin{enumerate}[(i)]
\setcounter{enumi}{1}
\item for any time $t$, we have
\[\mathcal{V}_t(X^t) = \mathcal{V}_t(\hat X), \qquad \mathbb{P}\text{-a.s.}\]
\end{enumerate}
\end{definition}

\begin{remark}

Unlike most interpretations of Bellman's principle, this definition is `forward looking', and does not, in general, admit the use of dynamic programming as a means of finding optimal policies.

Definition \ref{def:timeconsistent} directly allows initial behaviour to affect future behaviour in complex ways. This idea is embedded in the assumption of viability and the freedom to specify which policies are in $\mathcal{U}$. That is, this approach allows precommitment contracts and similar devices to be modelled, through restrictions on $\mathcal{U}$. The definition, therefore, looks for consistent behaviour \emph{contingent on what has already been done}. 

Essentially, if we choose an optimal policy today, we simply need to check that, in the future, we shall continue to follow a policy which we consider equivalent to the optimal choice today. In some sense, a policy is time consistent if it leads to a `\emph{commitment to previous decisions}'.

This definition has the distinct disadvantage of not requiring us to ensure that our decisions today will make us happy in the future. The policies selected as optimal in the future only need to lie in $\mathcal{U}|_t^{\hat X}$, that is, in the space of policies we have left ourselves to choose from.  A simple example of this is when the space $\mathcal{U}$ consists only of `buy-and-hold' policies. Here we make a decision at time zero, and are unable to modify it at any point in the future -- the policy $X^0$ chosen at time zero is the only policy in $\mathcal{U}|_{t}^{X^0}=\mathcal{U}|_{t}^{\hat X}$ for all $t>0$. Hence this decision is time consistent, as no deviation from the initial plan is permitted.
\end{remark}

\begin{remark}\label{rem:optiminequality}
It is important to note that, if $\{X^t\}$ is optimal, for any $t$, as $X^t$ maximises $\mathcal{V}_t$ and
\[\hat X \in \mathcal{U}|_t^{\hat X},\]
Property (ii) of Definition \ref{def:timeconsistent} can only ever fail through a future decision appearing sub-optimal today, that is, it is always true that
\[\mathcal{V}_t(X^t) \geq \mathcal{V}_t(\hat X), \qquad \mathbb{P}\text{-a.s.}\]
\end{remark}

\subsection{Bellman's principle and time-consistency}
We now give a relation between time consistency in the sense of Definition \ref{def:timeconsistent} and a type of intertemporal monotonicity for the value function. For simplicity, we write $[s,t[$ for the discrete collection of times $\{s,s+1,...,t-1\}$, and similarly for $]s,t]$.

This result is closely related to results of Artzner et al. \cite{Artzner2007}. Our approach differs from theirs mainly in the attention given to the space of possible policies $\mathcal{U}$.

\begin{theorem} \label{thm:equivtimeconsistent}
The following statements are equivalent
\begin{enumerate}[(i)]
\item The value $\mathcal{V}$ is such that every optimal policy choice is also time-consistent, for every initial compact policy set $\mathcal{U}$ of adapted processes in $\mathbb{U}$.

\item For any adapted processes $X, X'$ taking values in $\mathbb{U}$ and any times $s< t$, if $X_u = X'_u$ for all $u\in[0,t[$ and $\mathcal{V}_t(X)\geq \mathcal{V}_t(X')$ a.s. then $\mathcal{V}_s(X) \geq \mathcal{V}_s(X')$ a.s. 

\end{enumerate}
\end{theorem}
\begin{proof}
We interpret all (in-)equalities as $\mathbb{P}$-a.s.

\emph{(i implies ii.)} Assume our policy space is given by $\mathcal{U}=\{I_AX+I_{A^c}X': A\in\mathcal{F}_t\}$. Note that as $X_u=X'_u$ for all $u\in[0,t[$, we have $\mathcal{U}_t^{\hat X} = \mathcal{U}_t^{X} = \mathcal{U}_t^{X'}$. Then at time $t$, if $\mathcal{V}_t(X)\geq\mathcal{V}_t(X')$ we will find $X^t = X$ is an optimal policy. This implies that $\hat X = X$, as $X_u=X'_u$ for $u\in[0,t[$.  Hence by time consistency,
\[\mathcal{V}_{s}(X') \leq \mathcal{V}_s(X^s) = \mathcal{V}_s(\hat X)=\mathcal{V}_s(X).\]

\emph{(ii implies i.)}
Let $s$ be the first time that $\mathcal{V}_t(X^t) = \mathcal{V}_t(\hat X)$ for all $t>s$. By Lemma \ref{lem:XtnearT}, $s<T$. As $\{X^t\}$ is optimal and $X^{s+1}\in\mathcal{U}|_{s}^{X^s}$, we know
\[\mathcal{V}_{s+1}(X^{s})\leq \mathcal{V}_{s+1}(X^{s+1}) = \mathcal{V}_{s+1}(\hat X).\]
By (ii), this implies that $\mathcal{V}_{s}(X^{s}) \leq \mathcal{V}_{s}(\hat X).$
As $X^s$ is optimal, it follows that $\mathcal{V}_s(X^{s}) = \mathcal{V}_{s}(\hat X)$. 

Therefore, if $\mathcal{V}_t(X^t) = \mathcal{V}_t(\hat X)$ for all $t>s$, then $\mathcal{V}_s(X^s) = \mathcal{V}_s(\hat X)$. By induction, this must hold for all times, that is, the optimal choice is consistent.
\end{proof}

\begin{corollary}\label{cor:Bellmanconsistent}
The value function given by Bellman's equation is time consistent for \emph{any} initial policy set $\mathcal{U}$.
\end{corollary}
\begin{proof}
Let $f(\omega, s, X_s)$ be the payoff at time $s$ of following policy $X_s$. Bellman's equation then gives, for a fixed policy $X$, the value function
\[\mathcal{V}_s(X) = E[f(\omega, s, X_s) +\mathcal{V}_{s+1}(X)|\mathcal{F}_s],\]
By recursion, given past policy $\hat X$, this clearly implies that $\mathcal{V}_s(X)$ is a functional only of $\{X_u\}_{u\in[s,t[}$ and $\mathcal{V}_t(X)$, and hence, statement (ii) of Theorem \ref{thm:equivtimeconsistent} is satisfied.
\end{proof}

\begin{corollary}
The value functions given by dynamic risk measures and nonlinear expectations are time consistent for \emph{any} initial policy set $\mathcal{U}$.
\end{corollary}
\begin{proof}
In this context, the policy $X$ determines a (stochastic) terminal value $V_T^X$. Our nonlinear expectation/dynamic risk measure the yields the value 
\[\mathcal{V}_t(X):=\mathcal{E}(V_T^X|\mathcal{F}_t)=-\rho_t(V^X_T).\]
By the recursivity and monotonicity properties of nonlinear expectations/dynamic risk measures, we can write $\mathcal{V}_s(X)$ as a nondecreasing functional of the future values $\mathcal{V}_t(X)$. Hence statement (ii) of Theorem \ref{thm:equivtimeconsistent} is satisfied.
\end{proof}

\subsection{Dependable decisions}
We now propose a new type of `time consistency', which we call `dependability'. One can characterise classical time-consistency through the statement `a policy $X$ is time consistent if the policies chosen in the future, pasted together with $X$, give \emph{the same value today} as $X$ does.' 

Our new definition would then read, `a policy $X$ is \emph{dependable} if the policies chosen in the future, pasted together with $X$, give \emph{higher values today} than policy $X$ does'. In some sense, dependable policies are those that form a lower bound on the value function, irrespective of future decisions.

Intuitively, we suppose that, at any given time, a decision maker can only consider a subset of all possible plans, and will select the optimal policy from this subset. As time progresses, more plans can be considered, and so preferable alternatives may arise. `Dependability' is then a notion of time-consistency which allows for these new alternatives.

\begin{definition} \label{def:dependable}
Consider a standard problem. Suppose that, for each time $t\geq 0$, we only consider policies restricted to some compact subset $\tilde{\mathcal{U}}|_{t}^{\hat X}\subseteq \mathcal{U}|_{t}^{\hat X}$. Assume $\tilde{\mathcal{U}}|_{t}^{\hat X}$ satisfies the pasting property (\ref{eq:pastingproperty}).

Let $\{X^t\}$ be a viable collection of policies $X^t\in \mathcal{U}|_s^{\hat X}$ for $s<t$, such that $X^t\in\tilde{\mathcal{U}}|_t^{\hat X}$ for each $t$. Note in general $X^t\notin\tilde{\mathcal{U}}|_s^{\hat X}$ for $s<t$. 

This collection is called $\tilde{\mathcal{U}}$-\emph{optimal} if
\begin{enumerate}[(i)]
\item for any $t$, the policy $X^t$ uniformly maximises $\tilde{\mathcal{V}}_t(X)$ for $X\in\tilde{\mathcal{U}}|_{t}^{\hat X}$,
\end{enumerate}
and \emph{dependable} if
\begin{enumerate}[(i)]
\setcounter{enumi}{1}
\item for any time $t$, we have
\[\tilde{\mathcal{V}}_t(X^t) \leq \tilde{\mathcal{V}}_t(\hat X), \qquad \mathbb{P}\text{-a.s.}\]
\end{enumerate}
\end{definition}

\begin{remark}
As highlighted by Remark \ref{rem:optiminequality}, when $\tilde{\mathcal{U}}|_{t}^{\hat X} = \mathcal{U}|_{t}^{\hat X}$, this will degenerate into the usual time-consistency properties. Here, on the other hand, our restricted set of policies $\tilde{\mathcal{U}}|_{t}^{\hat X}$, over which we optimise at each time point, can make our problem time-inconsistent. 
\end{remark}
\begin{remark}
As $\tilde{\mathcal{U}}|_{t}^{\hat X}$ is compact and satisfies the pasting property (\ref{eq:pastingproperty}), the other `standard' properties of $\mathcal{V}$ show the existence of a policy $X^t$ uniformly maximising $\mathcal{V}_t$ on $\tilde{\mathcal{U}}|_{t}^{\hat X}$.

Note that as we have now restricted the set of policies which we can consider at any time point, the result of Theorem \ref{thm:equivtimeconsistent} no longer applies.
\end{remark}

Under this definition, it is perfectly reasonable that a na{\"\i}ve policy may be selected early on. However, when it is reconsidered later, this decision might be changed. The difference is that this decision is `dependable' if, had we been allowed to initially consider the decision with the later change, we would have preferred it to the policy initially chosen.

This `dependable' approach to time-consistency is a natural one for problems where only a subset of possible policies can be considered at each time. We shall see that the problems induced by the moving-horizon approach to risk measurement are of this type.

\section{An investment policy model}\label{sec:movhorizintro}
We now move to the specific problem of consistency of decisions based on a moving horizon.

Consider a probability space based on a classical model of a financial market in discrete time. We assume that all positions will be closed out at or before some distant deterministic time $T$. Hence, time can be indexed by the set $\{0,1,...,T\}$.

 We suppose that there are $d$ risky assets $\{S^i\}$ defined on some complete filtered probability space $(\Omega, \mathcal{F}, \{\mathcal{F}_t\}, \mathbb{P})$. We assume that $\mathbf{S}_t\in L^2(\mathcal{F}_t)$ for all $t$, where $\mathbf{S}$ denotes the vector of risky asset prices. We also assume the existence of a `risk-free' asset, however, for simplicity, we shall assume that the risk free interest rate is zero. Equivalently, we assume all quantities have been appropriately discounted. We assume that there are no transaction costs. 

A firm wishing to invest in this market has a range of self-financing policies available, which is a subset $\mathcal{U}$ of the adapted processes in $\mathbb{R}^d=\mathbb{U}$. We assume that $\mathcal{U}$ is a compact subset of adapted processes $\{X:X_t\in L^2(\mathbb{R}^d; \mathcal{F}_t)\}$. 

An investor's wealth process $V^X$ satisfies the stochastic difference equation
\begin{equation}\label{eq:VdynamicsinS}
V_{t+1}^X = \langle X_t, \mathbf{S}_{t+1}-\mathbf{S}_{t}\rangle + V_t^X.
\end{equation}
(The risk-free asset could also be included, but as we assume the risk-free interest rate is zero, it would not affect the dynamics of $V^X$.) For notational simplicity, we extend $V$ beyond time $T$ by setting $V_u^X=V_T^X$ for all $u>T$. 
Note that a policy $X_t$ describes the choice to be made under every contingency, and is not required to be Markovian or of feedback form.

We now state the following general definition, due to Peng (eg \cite{Peng1997}, \cite{Peng2004}).

\begin{definition}[Nonlinear Expectations]
A system of operators 
\[\mathcal{E}(\cdot|\mathcal{F}_t): L^1(\mathcal{F}_T) \to L^1(\mathcal{F}_t)\]
is called an $\mathcal{F}_t$-consistent nonlinear expectation if it satisfies the following properties.
\begin{enumerate}
\item (Monotonicity) If $Q\geq Q'$ $\mathbb{P}$-a.s. then $\mathcal{E}(Q|\mathcal{F}_t) \geq \mathcal{E}(Q'|\mathcal{F}_t)$, with $\mathcal{E}(Q|\mathcal{F}_t) = \mathcal{E}(Q'|\mathcal{F}_t)$ only if $Q=Q'$ $\mathbb{P}$-a.s.
\item (Constant invariance) For $Q\in L^1(\mathcal{F}_t)$, $\mathcal{E}(Q|\mathcal{F}_t) = Q$.
\item (Recursivity) For any $s\leq t$, $\mathcal{E}(\mathcal{E}(Q|\mathcal{F}_t)|\mathcal{F}_s) = \mathcal{E}(Q|\mathcal{F}_s)$ $\mathbb{P}$-a.s.
\item (Zero-one law) For any $A\in\mathcal{F}_t$, $I_A\mathcal{E}(Q|\mathcal{F}_t) = \mathcal{E}(I_A Q|\mathcal{F}_t)$.
\end{enumerate}
\end{definition}

\begin{remark}
In \cite{Cohen2008c} and \cite{Cohen2009a}, we have given a representation result for these operators in discrete time on finite horizons using the theory of BSDEs. These results give a complete description of nonlinear expectations (and the more general class of nonlinear evaluations) in this context. These results are not germane to the present work so we shall simply proceed by assuming that a nonlinear expectation is given.
\end{remark}

At each time $t$, an investor wishes to choose the `time-$t$-optimal' policy $X^t\in\mathcal{U}|_{t}^{\hat X}$. We shall model their decision as based on a value function given by a dynamic risk measure $\rho_t(\cdot)$, or equivalently, by an $\mathcal{F}_t$-consistent nonlinear expectation $\mathcal{E}(\cdot|\mathcal{F}_t)$. That is we have the following problem:

\begin{definition}
The \emph{simple moving horizon problem} with horizon $m$ is to find a viable policy choice $\{X^t\}$, where for each $t$, $X^t$ uniformly maximises 
\[\mathcal{V}_t(X)=\mathcal{E}(V^{X}_{t+m}|\mathcal{F}_t) = -\rho_t(V_{t+m}^{X})\]
for $X\in\mathcal{U}|_t^{\hat X}$.
\end{definition}

We emphasise at this point that we have chosen our value function such that the time-inconsistency in this problem arises purely because of the short horizon. The nonlinear expectation itself is time-consistent, in the sense of \cite{Artzner2007}. However, the nonlinear expectation is not being evaluated on the terminal values, which, as we shall see, leads to inconsistencies.

\subsection{An equivalence for policies}
We now show that this problem is, in general, equivalent to a dependable problem.

\begin{definition}
We define
\[I_{[0,t+m[}\mathcal{U}|_{t}^{\hat X}=\{I_{[0,t+m[}X|X\in\mathcal{U}|_{t}^{\hat X}\}.\]
For a given horizon $m$, we say that $\mathcal{U}$ is closed under truncation if, for all times $t$, all past policies $\hat X$,
\[I_{[0,t+m[}\mathcal{U}|_t^{\hat X} \subseteq \mathcal{U}|_t^{\hat X}.\]
\end{definition}

The following Lemma is trivial to prove, however forms the basis for the desired equivalence.
\begin{lemma} \label{lem:horizoniszeros}
At any time $t$, for any policy $X\in I_{[0,t+m[}\mathcal{U}|_{t}^{\hat X}$ we have the identity
\[\mathcal{V}_t(X) = \mathcal{E}(V_{t+m}^X|\mathcal{F}_t) = \mathcal{E}(V_T^X|\mathcal{F}_t)\qquad \mathbb{P}\text{-a.s.}\]
Furthermore, for any policy $X\in\mathcal{U}$, any time $t$, 
\[\mathcal{V}_t(X) = \mathcal{V}_t(I_{[0,t+m[}X) \quad \mathbb{P}\text{-a.s.}\]

\end{lemma}

\begin{definition}
The \emph{modified moving horizon problem} with horizon $m$ is to find a viable policy choice $\{X^t\}$, where for each $t$, $X^t$ uniformly maximises 
\[\tilde{\mathcal{V}}_t(X)=\mathcal{E}(V^{X}_{T}|\mathcal{F}_t) = -\rho_t(V_{T}^{X})\]
for $X\in I_{[0,t+m[}\mathcal{U}|_t^{\hat X}=:\tilde{\mathcal{U}}_t^{\hat X}$.
\end{definition}

\begin{remark}
We can now consider the `moving horizon problem' in two distinct ways. Either
\begin{itemize}
\item We take the value function $\mathcal{V}_t(X)=\mathcal{E}(V^X_{t+m}|\mathcal{F}_t)$, in which case we have a time-inconsistent problem, or
\item we take the value function $\tilde{\mathcal{V}}_t(X)=\mathcal{E}(V^X_T|\mathcal{F}_t)$, and then require that our selection $X^t$ must lie in the set $I_{[0,t+m[}\mathcal{U}|_t^{\hat X}$ for each $t$.
\end{itemize}
By Lemma \ref{lem:horizoniszeros}, we can assume, without loss of generality, that the policy $X$ which maximises $\mathcal{V}_t$ will lie in this set, and for all such policies we have $\mathcal{V}_t(X) = \tilde{\mathcal{V}}_t(X)$. 

That is, we can consider the moving horizon in terms of a restriction on the policy space, rather than in terms of evaluating the wealth process $V$ at the moving horizon. The values associated with each policy under these alternative approaches will be identical. 
\end{remark}

We can now give the following positive result for the moving horizon problem.

\begin{theorem} \label{thm:movinghorizondependable}
Any $\tilde{\mathcal{U}}$-optimal solution to the modified moving horizon problem is dependable.

By Lemma \ref{lem:horizoniszeros}, when $\mathcal{U}$ is closed under truncation, this will give the same values and policy choices at all times as when using a moving horizon.
\end{theorem}

\begin{proof}
For each $t$, we choose $X^t$ to maximise $\tilde{\mathcal{V}}_t(X^t) = \mathcal{E}(V^{X^t}_T|\mathcal{F}_t)$, for $X^t\in\tilde{\mathcal{U}}|_t^{\hat X}$. We also know that, for any $X\in\tilde{\mathcal{U}}|_{t+1}^{\hat X} = \tilde{\mathcal{U}}|_{t+1}^{X^t}$, we have $\tilde{\mathcal{V}}_{t+1}(X)\leq \tilde{\mathcal{V}}_{t+1}(X^{t+1})$. Hence, by the monotonicity of nonlinear expectations, we know that 
\[\tilde{\mathcal{V}}_t(X) \leq \tilde{\mathcal{V}}_t(X^{t+1}) \text{ for all }X\in\tilde{\mathcal{U}}|_{t+1}^{\hat X}.\]
Specifically, this implies 
\[\tilde{\mathcal{V}}_t(X^t) \leq \tilde{\mathcal{V}}_t(X^{t+1}).\]

Similarly, it follows that 
\[\tilde{\mathcal{V}}_{t-1}(X^{t-1}) \leq \tilde{\mathcal{V}}_{t-1}(X^t) \leq \tilde{\mathcal{V}}_{t-1}(X^{t+1}).\]
where the last inequality is again by monotonicity of nonlinear expectations.

By induction, this argument shows that for any times $s<t$, 
\[\tilde{\mathcal{V}}_s(X^s)\leq \tilde{\mathcal{V}}_s(X^t).\]
Hence, for all $s\leq T$, as by Lemma \ref{lem:XtnearT} $\hat X = X^T$, we have the result
\[\tilde{\mathcal{V}}_s(X^s)\leq \tilde{\mathcal{V}}_s(\hat X).\]
\end{proof}

\begin{remark}\label{rem:nocommitmentdependable}
Note that the requirement on $\mathcal{U}$ is that, in some sense, it does \emph{not} enforce commitment, specifically that one can always choose to `quit at the horizon', that is, to take the truncated policy $I_{[0,t+m[}X$.
\end{remark}

\section{A dependable but inconsistent example} \label{sec:dependinconsistentexample}
To demonstrate the usefulness of these results, we give a simple, if contrived, example of a situation where the moving horizon approach is inconsistent, but the equivalent approach using a modified policy space is dependable.

Suppose our market contains only one asset $S$. The policy space $\mathcal{U}$ consists of those processes $X$ of the form $X_u=I_{u<\sigma}$ where $\sigma$ is a stopping time. 

Let $T=3$, and suppose that values are given by the nonlinear expectation
\[\mathcal{E}(Q|\mathcal{F}_t) = -10\log E[e^{-Q/10}|\mathcal{F}_t].\]
This is evaluated on a horizon two periods from the present, that is, $m=2$, and 
\[\mathcal{V}_t(X) = \mathcal{E}\left.\left(V_{t+m}^X\right|\mathcal{F}_t\right).\]

Let $S$ follow a non-recombining binomial tree, with independent increments given by 
\[\begin{split}
S_0 &= 20\\
S_1-S_0 &= \begin{cases} 1 & \text{w.p. } 0.5\\-0.1 & \text{w.p. } 0.5\end{cases}\\
S_2-S_1 &= \begin{cases} 0.1 & \text{w.p. } 0.5\\-10 & \text{w.p. } 0.5\end{cases}\\
S_3-S_2 &= \begin{cases} 100 & \text{w.p. } 0.5\\-0.1 & \text{w.p. } 0.5\end{cases}\\
\end{split}\]
Here w.p. denotes `with probability'.

It is then easy to see that, in every state of the world $\omega$, the policy chosen at each time will be:
\[\begin{split}
X^0_t &= \begin{cases} 1 & t=0\\ 0 & t>0\end{cases}\\
X^1_t=X^2_t =X^3_t&= 1 \quad \text{a.s. for all }t\\
\end{split}\]
and therefore $\hat X_t =1$ a.s. for all $t$. Comparing these at time $0$, we have 
\[\mathcal{V}_0(X^0) = 0.1889 > -2.4926 = \mathcal{V}_0(\hat X),\]
and so our optimal solution is not time-consistent. 

On the other hand, at any time $t$,  the permitted polices allow the choice $X_u=0$ for $u>t$. That is, $\mathcal{U}$ is closed under truncation, in the sense of Theorem \ref{thm:movinghorizondependable}. Hence we know that this decision is dependable, under an equivalent value function. To show this empirically, we define 
\[\tilde{\mathcal{V}}_t(X) = \mathcal{E}\left.\left(V_T^X\right|\mathcal{F}_t\right)\]
and instead consider, at each time, policies in the restricted set 
\[I_{[0,t+m[}\mathcal{U}|_t^{\hat X}=\tilde{\mathcal{U}}_t^{\hat X}.\] 
On this set, from Lemma \ref{lem:horizoniszeros}, we know $\mathcal{V}_t(X)=\tilde{\mathcal{V}}_t(X)$, and that a policy which uniformly maximises $\mathcal{V}_t$ will lie in this set. We obtain exactly the same optimal policies, but have the values
\[\tilde{\mathcal{V}}_0(X^0) = 0.1889 < 0.4741 = \tilde{\mathcal{V}}_0(\hat X),\]
and so see that (given $X^1=X^2=X^3=\hat X$) our choice is dependable. Note that $\mathcal{V}_0(X^0)=\tilde{\mathcal{V}}_0(X^0)$, as expected.

\section{Acceptable decisions}
Often the question of interest is whether a policy $X$ is acceptable, that is, has a value above a critical level. (Equivalently, has a risk below a critical level.) For simplicity, we consider this decision at time $t=0$. Suppose that this critical value is given by $\mathcal{V}_0(0)$, the value associated with the `null' policy $X\equiv 0$. Our concern that this value function is time-inconsistent, hence we could plan, at time $t=0$, to follow policy $X$, but not follow through with it in the future. In this event, we require a guarantee that the truncated policy which we eventually follow, $\hat X$, yields an acceptable value today.

Now suppose that, for a policy $X$ under consideration, we define the space of available policies
\[\mathcal{U} = \{I_{[0,\tau[}X| \tau \text{ a stopping time}\}.\]
That is, we suppose that could change from following policy $X$ to the policy $0$ at any stopping time $\tau$. This policy space is clearly closed under truncation and is compact. It follows that the optimal policy choice in $\mathcal{U}$, using the modified moving horizon value $\tilde{\mathcal{V}}$, is dependable.

Therefore, if $X$ is acceptable, that is, $\tilde{\mathcal{V}}_0(X)=\mathcal{V}_0(X)\geq 0$, then we can be sure that the optimal realised policy $\hat X$ satisfies
\[\tilde{\mathcal{V}}_0(\hat X)\geq \tilde{\mathcal{V}}_0(X^0)\geq \tilde{\mathcal{V}}_0(X)\geq 0\]
where $X^0$ is the time-zero optimal policy in $\mathcal{U}$.
 For this reason, when it is possible for a position to be `sold off' at any time, we can be confident that future actions will not act to decrease the value/increase the risk assigned to a policy today, at least under an equivalent value function $\tilde{\mathcal{V}}$.
 
\section{Conclusions}
We have discussed the theory of time-consistency, and have given a definition for a new type of property, that of `dependability'. We have shown that, for a simple model of a financial market, under some assumptions on the allowable policies, the optimal decision reached using a moving horizon approach is equal to an optimal dependable decision using an equivalent value function. 

This result gives a partial justification for using a moving horizon approach in risk management. Assume that one can always decide to stop investing at the horizon, (that is, to take the policy $I_{[0,t+m[}X$). Then one can be sure that the optimal policy today, considering only a finite horizon, will only be improved by future decisions. 

This analysis still assumes that the underlying value function used is recursive up to the horizon, in particular, that it is an $\mathcal{F}_t$-consistent nonlinear expectation. This could be weakened to assuming that it is simply a nonlinear evaluation, and with appropriate adaptation of the arguments involved, we can also remove the assumption that interest rates are zero or deterministic. However, if the value function used is not recursive, for example, as with Coherent Value at Risk, these results would not apply. Essentially this is because these value functions introduce different types of time inconsistency, apart from the issues of moving horizons.

Given the extreme uncertainties that may be faced when attempting to model asset dynamics in the very long term, it may be appropriate to use a moving horizon approach. At the same time, if decisions involve commitment beyond the horizon, (and hence the policy space is not closed under truncation, in the sense of Theorem \ref{thm:movinghorizondependable}), consideration of the longer term is necessary.

\bibliographystyle{plain}  
\bibliography{../RiskPapers/General}
\end{document}